\documentclass[conference]{IEEEtran}

\usepackage{graphics}
\usepackage{floatflt,url}
\usepackage{amsmath,amssymb,amscd}
\usepackage{subfigure, epsfig}
\usepackage{latexsym}
\usepackage{multirow}
\usepackage[latin1]{inputenc}
\usepackage{makeidx}
\usepackage{gensymb}
\usepackage{hyperref}
\usepackage{textcomp}
\usepackage{color}
\usepackage{pstricks}
\usepackage{pstricks-add}
\usepackage{pst-plot}

\newcommand{\mc}{\mathcal}

\newcommand{\PPP}{\mathbb{P}}
\newcommand{\M}{\mathcal{M}}
\newcommand{\W}{\mathcal{W}}

\newcommand{\G}{\mathcal{G}}
\newcommand{\SSS}{\mathcal{S}}
\newcommand{\C}{\mathcal{C}}

\newcommand{\N}{\mathbb{N}}

\newcommand{\PP}{\mathcal{P}}

\newtheorem{theorem}{Theorem}
\newtheorem{definition}{Definition}
\newtheorem{proposition}{Proposition}

\begin{document}

\title{The price of re-establishing perfect, almost perfect or public monitoring in
games with arbitrary monitoring }

%
\author{\IEEEauthorblockN{Ma\"{e}l Le Treust\IEEEauthorrefmark{1} and
Samson Lasaulce\IEEEauthorrefmark{1}}
\IEEEauthorblockA{\IEEEauthorrefmark{1}Laboratoire des Signaux et Syst\`{e}mes,
CNRS  - Universit\'{e} Paris-Sud 11 - Sup\'{e}lec,
91191, Gif-sur-Yvette Cedex, France\\
Email: \{letreust\},\{lasaulce\}@lss.supelec.fr}}


\maketitle

\begin{abstract}
This paper establishes a connection between the
notion of observation (or monitoring) structure in game theory and
the one of communication channels in Shannon theory. One of the
objectives is to know under which conditions an arbitrary monitoring
structure can be transformed into a more pertinent monitoring structure.
To this end, a mediator is added to the game. The
objective of the mediator is to choose a signalling scheme that
 allows the players to have perfect, almost perfect or public monitoring
and all of this, at a minimum cost in terms of signalling. Graph
coloring, source coding, and channel coding are exploited to deal
with these issues. A wireless power control game
is used to illustrate these notions but the applicability of
the provided results and, more importantly, the framework of
transforming monitoring structures go much beyond this example.
\end{abstract}


\section{Introduction}
\label{sec:intro}

Observation or monitoring structures are omnipresent in games,
especially in dynamic games. Monitoring structures specify what the
players effectively observe. These observations allow a given player
to construct his private history, which is used, at a given instant,
as an input of a function defining his strategy. For instance,
observations may consist of action profiles (this is the case in
repeated games with perfect monitoring \cite{Sorin92} and
fictitious play \cite{Brown51}), arbitrary signals (this is the case in
repeated games with public signals \cite{Tomala98} and with an observation graph \cite{RT98}), or realizations
of the individual utility function (this is the case in stochastic
games between learning automata \cite{Sastry-94} and repeated game
with incomplete information \cite{GossnerHernandezNeyman06}). The problem is
that when players interact in a game with an arbitrary observation
structure, the possible outcomes might turn out to be unpredictable
and, even when they are, they might not have important properties
such as Nash equilibria. To be concrete, the characterization of
equilibrium utilities in repeated games with an arbitrary
observation structure is still an open problem \cite{RT04}.
In interactive situations where game theory is relevant like
distributed power control in wireless networks \cite{Lasaulce-Tutorial-09},
it is common that terminals do not observe the transmit power levels
of the other terminals \cite{LeTreustLasaulce(PowerControlRG)10,LeTreustTembineLasaulceDebbah10}.
 Being not able to predict all possible
operating points for such a network may cause a problem for the
network designer. In particular, ensuring the existence of efficient
Nash equilibria can be highly desirable when terminals implement
learning algorithms with partial observations \cite{Chandramouli-08}.

The above considerations show the importance of being able to
transform a given monitoring structure into a new one. But, how can
this be done? And at what price? This paper precisely falls in the
general framework which consist in proposing solutions to implement
such transformations and evaluating their cost in terms of
signalling. As far as the provided results are concerned, the
authors do not provide complete answers to these new questions.
Indeed, the scope of this paper is as follows.
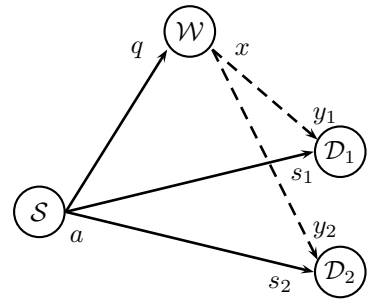
\begin{figure}[!ht]\label{figure:mediator}
\begin{center}
\psset{xunit=1cm,yunit=0.8cm}
\begin{pspicture}(-1,-1.3)(5,3.6)
\rput[u](2,3){$\mc{W}$} \rput[u](0,0){$\SSS$}
\rput[u](4,1){$\mc{D}_1$} \rput[u](4,-1){$\mc{D}_2$}
\rput[u](0.5,-0.4){$a$} \rput[u](2.7,2.7){$x$} \rput[u](1.3,2.7){$q$}
\rput[u](3.8,1.6){$y_1$} \rput[u](3.8,-0.3){$y_2$}
\rput[u](3.5,0.6){$s_1$} \rput[u](3.2,-1.2){$s_2$}
\pscircle(0,0){0.35} \pscircle(2,3){0.35} \pscircle(4,1){0.35}
\pscircle(4,-1){0.35} \psline[linewidth=1pt]{->}(0.35,0)(1.7,2.7)
\psline[linewidth=1pt]{->}(0.35,0)(3.65,1)
\psline[linewidth=1pt]{->}(0.35,0)(3.65,-1)
\psline[linestyle=dashed,linewidth=1pt]{->}(2.3,2.7)(3.7,1.2)
\psline[linestyle=dashed,linewidth=1pt]{->}(2.3,2.7)(3.7,-0.8)
\end{pspicture}
\caption{Interpreting the monitoring structure of a dynamic game as a communication problem.}
\end{center}
\end{figure}
First, one way to
transform a monitoring structure into a new one is to add a mediator (see Fig. \ref{figure:mediator})
in the game: this mediator does not have a strategic role here and is only
used for improving the observation capabilities of the players.
Second, even if the initial monitoring structure (without the
mediator) can be effectively arbitrary, the desired monitoring
resulting from the addition of the mediator is assumed to be perfect,
almost perfect or public, and therefore not arbitrary (the latter case is
left as a significant extension of this work). In the example of
distributed power control, the players would be the decisionnally
autonomous terminals while the mediator would be a base station or a
relay node. Whereas the ideas presented here seem seducing, the
question is how to tackle this general problem. One of the
contributions of this paper is to re-interpret observation
structures in games as channels in communication theory. Exploiting
this interpretation, several questions arise. Based on what the
mediator observes, does there exist a source code (at the mediator)
which allows the players to re-establish a perfect, almost perfect
or public observation of an information source (the action profiles
typically)? What is the minimum cost of signalling to re-establish
such an observation structure? Is the Shannon capacity
\cite{shannon-bell-1948} associated with the
initial observation structure high enough to convey the required
amount of signalling? Shannon theory \cite{cover-book-2006} and
graph theory \cite{BondyMurty}  brings appropriate answers to
all these questions. As it will be seen, the connection we establish
between game theory and Shannon theory opens many other interesting
issues such as: proving some equilibrium utilities are impossible to
reach in certain games because of limited channel capacities of the
considered observation structure; defining new channels in
communication theory from observation scenarios in game theory.

We provide a characterization of compatible monitoring structure
and a coding scheme that reconstruct $\varepsilon$-Perfect
Monitoring in Sec. \ref{SecAPMCommonMsg}.
After computing the price of re-establishing
the almost perfect monitoring ($\mathrm{PREEPM}$) we investigate
the reconstruction of Perfect Monitoring of the source in Sec. \ref{SecPublicMonitoring},
and the one-shot reconstruction of the almost Perfect Monitoring in Sec. \ref{SecAPMnodelay}.
We illustrate our results with the well-known ``prisoner's dilemma" in Sec. \ref{SecPrisonerAPM}.
The proof of the theorem are provided in the appendices \ref{Appendix}.

\section{System Model}

The purpose of this section is twofold: to review some basic
concepts and definitions from dynamic games, which are essential for
understanding the subsequent sections; to state
the general problem under investigation. Following the definition of Ba\c{s}ar
and Olsder (\cite{BasarOlsder82} pp. 205), a dynamic game consists in
 a sequence of stage games $\Gamma=(\mc{G}^t)_{t\in \N^*}$
 where at each stage $t\in \N^*$, we have:
 \begin{eqnarray*}
 \mc{G}^t&= &(\mc{K},\{\mc{P}^t_i\}_{i\in\mc{K} },\{\pi^t_i\}_{i\in\mc{K}},\omega^t,f^t,\\
&&\{\mc{S}^t_i\}_{i\in\mc{K} }, \{g_i^t\}_{i\in\mc{K}},\{h^t_i\}_{i\in\mc{K} }, \{\tau_i^t\}_{i\in\mc{K}})
\end{eqnarray*}
Denote $\mc{K} =\{1,...,K\}$ the set of players constant along the game,
 $\mc{P}^t_1,...,\mc{P}^t_K$ are the
corresponding sets of actions, $\pi^t_1,...,\pi^t_K$ are the payoff
(or cost) functions, $\omega^t$ is the state parameter and $f^t$
 is the state transition function, $g^t_1,...,g^t_K$
are the private monitoring functions at stage $t$ and $\mc{S}^t_1,...,\mc{S}^t_K$
are the corresponding sets of private signals, $h^t_1,...,h^t_K$
are the private histories and $\tau^t_1,...,\tau^t_K$
are the strategy functions.
Game stages correspond to time intervals at the beginning of which
players can choose their actions.
\begin{figure}[!ht]\label{figure:channel}
\begin{center}
\begin{tiny}
\psset{xunit=0.5cm,yunit=0.3cm}
\begin{pspicture}(-1,-6)(12,6)
\rput[u](-1.5,0){$\SSS$}
\rput[u](4.35,-5){$g_2$}
\rput[u](4.35,5){$g_1$}
\rput[u](10,2.5){$\mc{D}_1$}
\rput[u](10,-2.5){$\mc{D}_2$}
\rput[u](-0.6,-0.3){$a$}
\rput[u](10.3,3.5){$s_1$}
\rput[u](10.3,-3.6){$s_2$}
\rput[u](11.2,2.8){$\hat{a}_1$} \rput[u](11.2,-2.2){$\hat{a}_2$}
\psframe(-2,-0.5)(-1,0.5)
\psframe(9.5,2)(10.5,3)
\psframe(9.5,-3)(10.5,-2)
\pscircle(4.35,-5){0.35}
\pscircle(4.35,5){0.35}
\psline[linewidth=1pt](-1,0)(0,0)
\psline[linewidth=1pt](0,0)(0,-5)
\psline[linewidth=1pt](0,0)(0,5)
\psline[linewidth=1pt](0,-5)(3.65,-5)
\psline[linewidth=1pt](5.05,-5)(10,-5)
\psline[linewidth=1pt](0,5)(3.65,5)
\psline[linewidth=1pt](5.05,5)(10,5)
\psline[linewidth=1pt]{->}(10,5)(10,3)
\psline[linewidth=1pt]{->}(10,-5)(10,-3)
\psline[linewidth=1pt]{->}(10.5,2.5)(11.5,2.5)
\psline[linewidth=1pt]{->}(10.5,-2.5)(11.5,-2.5)
\end{pspicture}
\end{tiny}
\caption{The private monitoring channel.}
\end{center}
\end{figure}
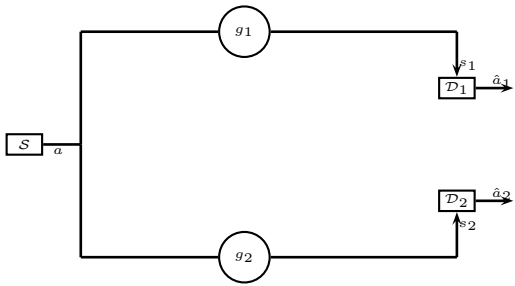

The strategic information is modeled by an information source
where $a(t)$ is produced by the
source at stage $t$. This strategic information may consists in action
profiles or arbitrary signals. We assume that, for a given game stage
$t \geq 1$, each player $i\in \mc{K}$ knows
and can take into account the past realizations of his private observation  $s_i$
drawn from the private monitoring $g_i$ (see Fig. \ref{figure:channel}).
Denote $\Delta(Z)$  the set of probabilities over the set $Z$.
\begin{eqnarray}
g_i : \mc{A} \longrightarrow \Delta(\mc{S}_i)
\end{eqnarray}
The main difference between static games and dynamic games is that players can take into account the
sequence of past strategic signals in their long-run strategy. Increasing the amount of strategic
information, increase the strategy space of the players. The vector $h^t_i = (s_i(1), ...,s_i(t-1))$
 is the private history of player $i$, at stage $t$ and lies in the
set $\mc{H}^t_i  = \left(\mc{S}_i \right)^{t-1}$. A strategy $\tau_i$ for player $i \in \mc{K}$ is a sequence of
strictly causal functions $\left(\tau_{i,t} \right)_{t \geq 1}$,
\begin{eqnarray}
\tau_{i,t}: \mc{H}^t_i  \rightarrow   \mc{P}^t_i
\end{eqnarray}
Let $\mc{T}_i$ be the set of strategies $\tau_i$ of player $i\in\mc{K}$ and $\tau =
(\tau_1, ..., \tau_K)$ be a joint strategy.

We introduce an additive signalling structure  called
``the mediator assisted monitoring channel", represented
 in Fig. (3).
 It consist of a triple $(\mc{W},m,f)$ where $\mc{W}$ denote the mediator, $m$ the observation channel of the mediator and $f$ the communication channel between the mediator and the players. The mediator also observes a noisy version $q$ of the information source $a$. It's has to relay every relevant information to the players in order to make them monitors the information source. The observation channel of the mediator is defined as follows.
Denote  $\mc{A}$ the set of strategic information and $Q$ the set of signals observed by the mediator.
\begin{eqnarray}
m : \mc{A} \longrightarrow \Delta(Q)
\end{eqnarray}
The communication  channel between the mediator and the players where $X$
 is the set of channel inputs and $Y_i$ is the set of signals observed by player $i\in \mc{K}$.
\begin{eqnarray}
f : X \longrightarrow \Delta(Y_1\times Y_2)
\end{eqnarray}
Thus at each stage $t\geq 1$  of the game, the players obtain a private observation $s_i^t$ and
a mediator's signal $y_i^t$. We investigates the properties of such an additive signalling structure in order to answer the question: Are the players able to observes the information source or not ?

Denote $\hat{a}^t_i$ the reconstructed version of the source by player $i\in\mc{K}$.
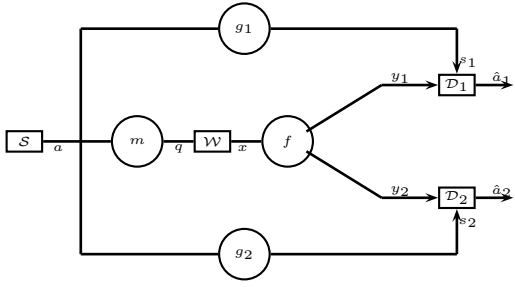
\begin{figure}[!ht]\label{figure:medcha}
\begin{center}
\begin{tiny}
\psset{xunit=0.5cm,yunit=0.3cm}
\begin{pspicture}(-1,-6)(12,6)
\rput[u](-1.5,0){$\SSS$} \rput[u](1.5,0){$m$}
\rput[u](4.35,-5){$g_2$} \rput[u](4.35,5){$g_1$}
\rput[u](3.5,0){$\mc{W}$} \rput[u](10,2.5){$\mc{D}_1$}
\rput[u](10,-2.5){$\mc{D}_2$} \rput[u](-0.6,-0.3){$a$}
\rput[u](2.6,-0.3){$q$} \rput[u](4.3,-0.3){$x$}
 \rput[u](5.5,0){$f$}
\rput[u](8.5,2.8){$y_1$} \rput[u](8.5,-2.2){$y_2$}
\rput[u](10.3,3.5){$s_1$} \rput[u](10.3,-3.6){$s_2$}
\rput[u](11.2,2.8){$\hat{a}_1$} \rput[u](11.2,-2.2){$\hat{a}_2$}
\psframe(-2,-0.5)(-1,0.5)
\psframe(3,-0.5)(4,0.5)
\psframe(9.5,2)(10.5,3)
\psframe(9.5,-3)(10.5,-2)
\pscircle(1.5,0){0.35}
\pscircle(5.5,0){0.35}
\pscircle(4.35,-5){0.35}
\pscircle(4.35,5){0.35}
\psline[linewidth=1pt]{->}(10.5,2.5)(11.5,2.5)
\psline[linewidth=1pt]{->}(10.5,-2.5)(11.5,-2.5)
\psline[linewidth=1pt](-1,0)(0.8,0)
\psline[linewidth=1pt](2.2,0)(3,0)
\psline[linewidth=1pt](4,0)(4.8,0)
\psline[linewidth=1pt](6,0.44)(8,2.5)
\psline[linewidth=1pt](6,-0.44)(8,-2.5)
\psline[linewidth=1pt]{->}(8,2.5)(9.5,2.5)
\psline[linewidth=1pt]{->}(8,-2.5)(9.5,-2.5)
\psline[linewidth=1pt](0,0)(0,-5)
\psline[linewidth=1pt](0,0)(0,5)
\psline[linewidth=1pt](0,-5)(3.65,-5)
\psline[linewidth=1pt](5.05,-5)(10,-5)
\psline[linewidth=1pt](0,5)(3.65,5)
\psline[linewidth=1pt](5.05,5)(10,5)
\psline[linewidth=1pt]{->}(10,5)(10,3)
\psline[linewidth=1pt]{->}(10,-5)(10,-3)
\end{pspicture}
\end{tiny}
\caption{The mediator-assisted monitoring channel.}
\end{center}
\end{figure}%
The course of the signalling process begins with the strategic
information $a$, generated by the source at a given stage.
The mediator $\mc{W}$ is assumed to have an imperfect
observation (namely $q$) of the symbols $a$ generated by the source
and knows the information structure of every player. Taking this
knowledge into account, the mediator applies certain mathematical operations
on what it observes and broadcasts a public signal $x$ to all the
players. Therefore, each player $i\in \mc{K}$
 receives a private signal $s_i$ and an additional
signal from the mediator denoted by $y_i$.

%

\section{Reconstruction of the $\varepsilon$-Perfect Monitoring}\label{SecAPMCommonMsg}

In this section, we investigate the reconstruction of the $\varepsilon$-Perfect Monitoring.
We introduce an additive signalling structure $(\mc{W},m,f)$ which operates as a relay in order to send an additional signal to the players.
We provide conditions over the additive signalling structure in order the players monitors almost perfectly the source of strategic information.
We first recall the definition of $\varepsilon$-perfect monitoring available in the literature
 \cite{ElyValimaki(Robust)02,HornerOlszewski06} and we present a ``max-min formulation" to compute the error parameter $\varepsilon$.
Then we define properly the ``reconstruction" of the $\varepsilon$-perfect monitoring at the players.
Based on a graph-coloring approach, we provide two conditions over the additive signalling structure $(\mc{W},m,f)$ that are sufficient to reconstruct the $\varepsilon$-perfect monitoring for the players. We call the first condition: ``the $(x,y)$-coloring condition". It regards the observation function of the mediator $m$ and it guarantee that the mediator can reconstruct the $\varepsilon$-perfect monitoring.
The second condition concerns the communication channel $f$ between the mediator and the players and is called the ``essential information condition".
It guarantee that the capacity of the channel $f$ allows the mediator to communicate the strategic information to the players.

In this section, we investigate the reconstruction problem using the framework of Shannon \cite{shannon-bell-1948}. We  make the following assumptions on the information source, the private monitoring and the mediator assisted channel.
\begin{itemize}
\item The information source is discrete and i.i.d.
\item The monitoring structure is stationary.
\item The players may tolerate a delay in the signalling.
\end{itemize}
These assumption allow us to derive the fundamental limit derived by Shannon on the information transmission. The results, we present in this section, are based on the three above  assumptions. However, the  strategies of the players may not always satisfy this properties. We relax these three hypothesis in Sec. (\ref{SecAPMnodelay}) and we derive alternative limits over the information transmission.

\begin{definition}{\cite{ElyValimaki(Robust)02,HornerOlszewski06}}
A monitoring $\Lambda : A \longrightarrow \Delta(\prod_{i\in \mc{K}}\Sigma_i)$ is $\varepsilon$-perfect (or almost perfect) if for each player $i\in \mc{K}$ there exists a partition $T_i=\{T_i^a: a \in \mc{A}\}$ of the signals $\Sigma_i$ such that for all $a\in \mc{A}$,
\begin{equation}
\sum_{\sigma_i\in T_i^a} \Lambda(\sigma_i|a)\geq 1-\varepsilon
\end{equation}
\end{definition}
We characterize the precision of the monitoring using a ``max-min formulation".
\begin{proposition}
A monitoring $\Lambda$ is $\varepsilon$-perfect if and only if
\begin{eqnarray*}
1-\varepsilon&=&\min_{i\in \mc{K}}\max_{T_i=(T_i^a)_a}\min_{a\in \mc{A}}\sum_{\sigma_i\in T_i^a} \Lambda(\sigma_i|a)\\
\Longleftrightarrow \varepsilon&=&\max_{i\in \mc{K}}\min_{T_i=(T_i^a)_a}\max_{a\in \mc{A}}\sum_{\sigma_i\notin T_i^a} \Lambda(\sigma_i|a)
\end{eqnarray*}
\end{proposition}
\begin{proof}
See Appendix \ref{ProofProposition1}
\end{proof}

After a joint action $a$ is played, each player $i\in \mc{K}$ obtains a private signal $s_i$ drawn from a private monitoring $g_i$.
\begin{eqnarray}
g_i &:& A \longrightarrow \Delta(S_i)\quad \forall i\in  K
\end{eqnarray}
The mediator observes a signal $q$ drawn from the observation channel $m$.
\begin{eqnarray}
m &:& A\longrightarrow \Delta(Q)
\end{eqnarray}
Then it send through the communication channel $f$ an additive signal to each players.
\begin{eqnarray}
f &:& X\longrightarrow \Delta(\prod_{i\in \mc{K}} Y_i)
\end{eqnarray}
This communication procedure induces a pair of signals $\sigma_i=(s_i,y_i)$ for each player where $s_i$ comes from the private monitoring $g_i$ and $y_i$ comes from the additional signalling structure $(\mc{W},m,f)$. We derive conditions over the additional signalling structure $(\mc{W},m,f)$ such that the joint signal $\sigma_i=(s_i,y_i)$ satisfies the $\varepsilon$-perfect condition.

\subsection{Reconstruction of the $\varepsilon$-Perfect Monitoring}

We define the notion of code in this framework. The mediator observes a sequence of signals $q$ and reduce it to what we called the ``essential information sequence" $r$ using a graph coloring argument. Then it encodes the sequences $r$ into a sequence $x$ using a joint source-channel coding procedure.

The players will decode the ``essential information" $r$ using the channel output $y_i$ and the private observation $s_i$. This ``essential information sequence" $r$ combined with the sequence of private monitoring $s_i$ characterizes a unique sequence $a$ of joint actions.
\begin{definition}
A $(n,h,\phi,(\psi_i)_{i\in \mc{K}})$-code is a pair of encoding functions for the mediator:
\begin{eqnarray*}
&h& : Q \longrightarrow R,\qquad \text{``essential information"}\\
&\phi& :  R ^n \longrightarrow X^n,\qquad \text{``source-channel encoding"}
\end{eqnarray*}
and a decoding function for each player:
\begin{eqnarray*}
&\psi_i& : Y_i^n \times S_i^n \longrightarrow A^n, \forall i\in \mc{K},\; \text{``source-channel decoding"}
\end{eqnarray*}
\end{definition}

We quantify the precision of the joint signal $\sigma_i=(s_i,y_i)$ using the following definition.
\begin{definition}
The mediator can reconstruct the $\varepsilon$-Perfect Monitoring if,
\begin{eqnarray*}
&\forall \delta>0,\;\exists (n,h,\phi,(\psi_i)_{i\in \mc{K}})\text{-process such that},&\\
&\PP\left[\exists \{T_i^a\}_{a \in A},\;\forall a\in A,\; \sum_{\sigma_i\in T_i^a} \Lambda(\sigma_i|a)\geq 1-\varepsilon\right]\geq 1-\delta&
\end{eqnarray*}
\end{definition}

For a given private monitoring structure $(g_i)_{i\in K}$, we provide sufficient conditions over the additive signalling structure $(\mc{W},m,f)$ such that the mediator can reconstruct the $\varepsilon$-perfect monitoring. Two natural questions arises : When the mediator observation function $m$ is sufficiently precise to guarantee  the $\varepsilon$-perfect monitoring at the players ? When the communication channel  $f$ between the mediator and the players allows to transmit all the relevant information ?

We provide an answer to the first question  using the $(x,y)$-coloring condition in the next subsection (\ref{subsec:xy-colorringcond}). The second question will be investigate in subsection (\ref{subsec:rate-ess-info}) using the concept of ``rate of essential information".

\subsection{The $(x,y)$-coloring Condition}\label{subsec:xy-colorringcond}

We define the $(x,y)$-coloring condition in order to characterize the observation functions $m$ of the mediator that are compatible with every private monitoring $g_i$ of the players $i\in \mc{K}$. This condition is based on a graph-coloring approach. We represent the private monitoring $g_i$ using an auxiliary graph (see Def. \ref{def:auxiliarygraph}) whose vertices are the joint actions $a$. There is an edge $e$ between two vertices $a$ and $a'$ if both joint action induce the same signal $s_i$ with large probability.

The main idea is the following. If the observation of the mediator $m$ is a coloring of the auxiliary graphs, then the information $m$ passing through the mediator is completely orthogonal to the private information $g_i$. Thus every joint actions can be distinguished by the players and the $\varepsilon$-perfect monitoring can be reconstructed.

\begin{definition}
Define the equivalence classes of actions for each of the private monitoring $g_i$ with
$i\in \mc{K}$ as follows.
\begin{eqnarray}
G_i(a)&=&\{s_i\in S_i,\; g_i(s_i|a)>1/2\},\\
a &\sim_{g_i}& b \Longleftrightarrow G_i(a) = G_i(b)
\end{eqnarray}
Denote $A_{g_i}=\{\alpha_i\}$ the partition of $A$ into equivalence
classes with respect to the relation $\sim_{g_i}$.
In the same way with the monitoring $m$.
\begin{eqnarray}
M(b)&=&\{q\in Q,\; m(q|b)>1/2\},\\
a &\sim_m& b \Longleftrightarrow M(a) = M(b)
\end{eqnarray}
Denote $A_m=\{\alpha_m\}$ the partition of $A$ into equivalence
 classes with respect to the relation $\sim_m$.
These equivalence classes induce a family of auxiliary monitoring defined by.
\begin{eqnarray}
\tilde{g_i} : A_{g_i} &\longrightarrow& \Delta(S_i)^{|A|}\\
\alpha_i &\longrightarrow&  (g_i(s|a))_{a\in \alpha_i}
\end{eqnarray}
and
\begin{eqnarray}
\tilde{m} : A_m &\longrightarrow& \Delta(Q)^{|A|}\\
\alpha_m &\longrightarrow&  (m(q|a))_{a\in \alpha_m}
\end{eqnarray}
The precision of the auxiliary monitoring $\tilde{g_i}$ and $\tilde{m}$ are computed in the following way. Let $\{S_{\alpha}\}_{\alpha\in A_{g_i}}$ a partition of the signals $s$ of player $i$ indexed by the equivalence classes $\alpha\in A_{g_i}$. Define in the same way $\{Q_{\beta}\}_{\beta\in A_m}$ a partition of the signals $q$ of mediator indexed by the equivalence classes $\beta\in A_m$.
\begin{eqnarray}
&&\max_{S_{\alpha}}\min_{\alpha\in A_{g_i}}\min_{a\in \alpha}\sum_{s\notin S_{\alpha}} g_i(s|a) = x_i\\
&&\max_{Q_{\beta}} \min_{\beta\in A_m}\min_{a\in \beta}\sum_{q\notin Q_{\beta}} m(q|a) = y
\end{eqnarray}
The monitoring $\tilde{g_i}$ is $x_i$-perfect and $\tilde{m}$ is $y$-perfect.
\end{definition}

\begin{definition}\label{def:auxiliarygraph}
The auxiliary graph of player $i\in \mc{K}$, denoted $\G_i = (A,E_i)$ is defined as follows,
\begin{eqnarray}
\exists e_i = (a,b)\in E_i \Longleftrightarrow a\sim_{g_i} b
\end{eqnarray}
\end{definition}
Inspired from graph coloring we define the following concept of $(x,y)$-coloring.
\begin{definition}
The monitoring $g_i$ and $m$ satisfy an $(x,y)$-coloring condition if :
\begin{itemize}
\item The auxiliary monitoring $\tilde{g_i}$ is $x$ perfect,
\item The auxiliary monitoring $\tilde{m}$ is $y$ perfect,
\item The partition $\{Q_{\beta}\}_{\beta\in A_m}$ induced by the auxiliary monitoring $\tilde{m}$ is a coloring $c : A \longrightarrow Q$ of the graph $\G_i $.
\end{itemize}
Remark that the last condition is equivalent to the following one:
the auxiliary monitoring $\tilde{g_i}$ is a coloring of the graph $\G_m $
defined by $e_m=(a,b)\in E_m \Longleftrightarrow a\sim_m b$.
\end{definition}

\subsection{The Rate of Essential Information}\label{subsec:rate-ess-info}
We define the rate of essential information in order to characterize
the channels $f$ between the mediator and the players that are compatible
with the amount of information the players need.
It could happened that the observation channel $m$ of the mediator satisfy
the above $(x,y)$-coloring condition, but not all the information $q$ is relevant.

In this subsection, we aim at reducing the relevant information to it's minimum.
To do so, we use a second coloring condition over a bi-auxiliary
graph $\widetilde{\G}$ to eliminate any redundant information between the signals $q$ and $s_i$.
We call the ``essential information" the sequence $r$ corresponding to a concatenation of the sequence of signals $q$.

\begin{definition}
The bi-auxiliary graph $\widetilde{\G} = (Q,\widetilde{E})$ is defined as follows,
\begin{eqnarray}
\exists e = (q,q')\in \widetilde{E} &\Longleftrightarrow& \exists i\in \mc{K},\; \exists a,b\in A, \text{ s.t. }\\
 &&q\in m(a),\\
 && q'\in m(b),\\
 && a \sim_{g_i} b
\end{eqnarray}
\end{definition}

\begin{definition}
Let $\tilde{h}: Q \longrightarrow R$ the minimal coloring of the bi-auxiliary graph $\widetilde{\G}$ and denote the random variable $r$ essential information drawn from the distribution $\tilde{h}\otimes m \otimes p$ such that $P(r)=\sum_{a,q} p(a)m(q|a)\tilde{h}(r|q)$.\\
Define the essential rate as follows.
\begin{eqnarray}
H=\max_{i\in \mc{K}}H(r|s_i)
\end{eqnarray}
where the random variable $s_i$ is drawn from the transition $T_i : R \longrightarrow \Delta(S_i)$ with,
\begin{eqnarray}
T_i(s|r) &=& \frac{\sum_{a,q}\PPP(a,q,r,s)}{\sum_{a,q}\PPP(a,q,r)} \\
& =& \frac{\sum_{a,q}p(a)m(q|a)\tilde{h}(r|q)g_i(s|a)}{\sum_{a,q}\sum_{a,q}p(a)m(q|a)\tilde{h}(r|q)}
\end{eqnarray}
\end{definition}
In the following, such a mapping $h$ is called recoloring of monitoring $m$.
The following coding theorem for broadcast channel with common messages
\cite{korner-it-1977} provides us an upper bound for transmits to the players the strategic information.

\begin{theorem}[Korner, Marton 1977 \cite{korner-it-1977}]
The capacity $\C_0$ of the broadcast channel $f : X\longrightarrow \Delta(\prod_{i\in \mc{K}} Y_i)$  with common messages is exactly,
\begin{eqnarray}
\C_0 = \max_{p\in \Delta(X)}\min_{i\in \mc{K}} I(X;Y_i)
\end{eqnarray}
\end{theorem}

The coding theorem we present is constructed over large blocs of strategic signals. Its implies that the players may tolerate a delay in the reconstruction of the $\varepsilon$-perfect monitoring. This assumption is relaxed in section (\ref{SecAPMnodelay}) below and an alternative result is presented.

\subsection{Main Result}\label{subsec:main-result}

We provide two conditions that ensure the additive signalling
structure $(\W,m,f)$ is compatible with the reconstruction
of the  $\varepsilon$-perfect monitoring. The first condition is based on
the $(x,y)$-coloring condition (see subsection (\ref{subsec:xy-colorringcond})) and guarantees that the
mediator is sufficiently informed to help the players
reconstruct the desired monitoring. The second condition is based on the
 ``essential information" (see subsection (\ref{subsec:rate-ess-info})) and ensures that the additional
information the mediator obtains, is compatible with the communication constraints
of the channel between the mediator and the players.

Condition $(1)$ : There exists a pair $(x,y)$ such that $x+y-xy\leq \varepsilon$ and for each player $i\in \mc{K}$, the private monitoring $g_i$ and the monitoring of the mediator $m$ satisfy an $(x,y)$-coloring condition.\\
Condition $(2)$ :  The essential rate $H$ satisfy $H\leq \C_0$, the capacity $\C_0$ of the channel $f$ with common messages.

\begin{theorem}[$\varepsilon$-PM]\label{theo:APM}
Fix a strategy profile $p\in \Delta(A)$, a monitoring structure $\M=(m,(g_i)_{i\in \mc{K}},f)$ and an $\varepsilon>0$.\\
If the monitoring structure $\M$ satisfy conditions $(1)$ and $(2)$, then the mediator can reconstruct the $\varepsilon$-Perfect Monitoring.
\end{theorem}

\begin{proof}
The proof is detailed in Appendix \ref{ProofTheorem2}.
\end{proof}

We provide conditions over the additive signalling structure $(\mc{W},m,f)$ that are sufficient to reconstruct the $\varepsilon$-perfect monitoring for the players. Note that a complete characterization is not available due to the problem of characterizing the precision of a two parallel monitoring functions.

We obtain a set of admissible additive signalling structure $(\mc{W},m,f)$ and we need an evaluation method to choose the best admissible additive signalling structure $(\mc{W},m,f)$ in term of signalling cost. For that reasons, we introduce the price of re-establishing $\varepsilon$-perfect monitoring as the ratio between the number of bits of the additive signalling and the number of bits of the source of strategic information.

\begin{definition}
Define the price of re-establishing $\varepsilon$-Perfect Monitoring:
\begin{equation}
\mathrm{PREEPM}^{\infty}(\varepsilon) = \frac{\max_{i\in \mc{K}}H(R|S_i)}{H(A)}
\end{equation}
\end{definition}
The worst case correspond to the situation where the mediator directly send the entire sequence of joint actions $a$. In that case the price is equal to 1. Obviously, the players would have all the strategic information and they can reconstruct the monitoring perfectly. However, this situation is not very interesting from our point of view since the capacity constraints between the mediator and the players may forbid the transmission of the strategic information.

Finding the minimal price of re-establishing $\varepsilon$-perfect monitoring is equivalent to finding the optimal admissible additive signalling structure $(\mc{W},m,f)$.

\section{Reconstruction of the Public Monitoring}\label{SecPublicMonitoring}

The problem of strategic observation are well studied in the framework of repeated game with public monitoring.
In this section, we assume that the source of strategic information is no more a joint action but a public signal.
For example, if the public signal we consider satisfies the ``individual
and pairwise full rank conditions" of \cite{FudenbergYamamoto09},
then the set of the equilibria is fully characterized even if the game is stochastic.
We extend our results to the perfect reconstruction of the information source without error (i.e. where $\varepsilon=0$).
We provide sufficient and necessary conditions on the additional signalling structure
 $\W$ for being compatible with the reconstruction of the perfect monitoring.

\subsection{The ``Painting" Condition}
The main difference here is the precision of the monitoring of the information source: $\varepsilon=0$.
We provide here a necessary and sufficient condition over the observation function $m$ of the mediator such as reconstruct the perfect monitoring of the source of strategic information. This condition is also based on graph coloring and we called it ``the painting condition" in reference to C. Berge.

We construct a graph where the vertices are the public signals $a$. There is an edge between to publics signals $a$ and $a'$ if the same private signal $s_i$ is drawn with positive probability. We prove that the observation of the mediator is orthogonal to the private monitoring if and only if the observations $q$ of the mediator is a coloring of the graph.
\begin{definition}
Denote the sets of possible signals.
\begin{eqnarray}
G_i(a)&=&\{s_i\in S_i,\; g_i(s_i|a)>0\},\quad \forall i\in \mc{K}\\
M(b)&=&\{q\in Q,\; m(q|b)>0\}
\end{eqnarray}
\end{definition}

\begin{definition}
The auxiliary graph of player $i\in \mc{K}$, denoted $\G_i = (A,E_i)$ is defined as follows:
\begin{eqnarray}
\exists e_i = (a,b)\in E_i \Longleftrightarrow G_i(a)\cap G_i(b)\neq \emptyset
\end{eqnarray}
\end{definition}

We define the concept of painting of a graph $\G$ as a correspondence $m : A \rightrightarrows Q$ if every selection $\bar{m}:A \rightarrow Q$ of $m$ is a coloring of the graph $\G$.

\begin{definition}
The monitoring of the mediator $m$ is a painting of the family of graphs $(\G_i)_{i\in \mc{K}}$ induced by the private monitoring $(g_i)_{i\in \mc{K}}$ if for all $i\in \mc{K}$ we have
\begin{eqnarray}
\exists e_i=(a,b) \in E_i \Longleftrightarrow m(a)\cap m(b)=\emptyset
\end{eqnarray}
\end{definition}

\subsection{Main Result}
As in the previous section, we provide two conditions (over $m$ and $f$) such that the additive signalling structure $(\W,m,f)$ is compatible with the reconstruction of the perfect monitoring. This result is stronger than the previous one because we provide necessary and sufficient conditions.

\begin{definition}
Define the following conditions:\\
Condition $(1')$ :  The monitoring of the mediator $m$ is a painting of the family of graphs $(\G_i)_{i\in \mc{K}}$.\\
Condition $(2)$ :  The essential rate $H$ satisfies $H\leq \C_0$, the capacity $\C_0$ of the channel $f$ with common messages.
\end{definition}

\begin{theorem}[PM]\label{theo:PM}
Fix a strategy profile $p\in \Delta(A)$ and monitoring structure $\M=(m,(g_i)_{i\in \mc{K}},f)$.\\
The mediator can reconstruct the Perfect Monitoring for Strategy $p$ if and only if the monitoring structure $\M$ satisfy conditions $(1')$ and $(2)$.
\end{theorem}

\begin{proof}
The proof is detailed in Appendix \ref{ProofTheorem3}.

\end{proof}
We obtain a set of admissible additive signalling structure $(\mc{W},m,f)$ and we introduce the price of re-establishing perfect monitoring in order to evaluate the performance of the reconstruction.

\begin{definition}
Define the price of re-establishing Perfect Monitoring:
\begin{equation}
\mathrm{PRPM}^{\infty}(\varepsilon) = \frac{\max_{i\in \mc{K}}H(R|S_i)}{H(A)}
\end{equation}
\end{definition}

Finding the minimal price of re-establishing $\varepsilon$-perfect monitoring is equivalent to finding the optimal admissible additive signalling structure $(\mc{W},m,f)$ for reconstruct the perfect monitoring of the information source.

\section{One-shot Reconstruction of the $\varepsilon$-Perfect Monitoring}\label{SecAPMnodelay}

In the previous sections, we have assumed that the source of strategic
 information  was i.i.d., the channel was stationary and the players tolerate a
delay before reconstructing the desired monitoring.
In this section, we relax these three hypothesis and we
investigate a ``one-shot" reconstruction of the $\varepsilon$-perfect monitoring.
Note that, the techniques we develop in this section also
apply to the reconstruction of the perfect monitoring.

Once the strategic information is drawn, the mediator
provides an additional information to the players before the end of the game stage.
The definition of the one-shot reconstruction consists in replacing the
number $n$ of stages by 1 in the definition of the
long term reconstruction in Sec.\ref{SecAPMCommonMsg}.

\subsection{The Condition $z$-Perfect}

The main difference regards the communication channel $f$ between the mediator and the players.
Assuming the ``one-shot" reconstruction prevent us to use the classical coding scheme from Shannon theory.
We introduce the condition $z$-perfect in order to characterize
the channels $f$ (see Fig. 3) between the mediator and the players compatible with
 the one shot reconstruction of the $\varepsilon$ perfect monitoring.

\begin{definition}
 The channel between the mediator and each player is $z$-Perfect if, for each players,
  there exists a partition $Y_i=\{Y_i^r\}_{r \in R}$ of the signals indexed by the
   set of essential information $R$ such that for all $r\in R$:
\begin{equation}
\sum_{y_i\in Y_i^r} f(y_i|r) \geq 1 -z
\end{equation}
\end{definition}

\subsection{Main Result}
The result of this section is widely based on the one in section (\ref{SecAPMCommonMsg}).
Define as above the bi-auxiliary graph $\widetilde{\G}$, and the mapping
$\tilde{h} : Q \longrightarrow R$ is called recoloring of monitoring $m$.
The condition (2) with the entropy inequality of theorem (\ref{theo:APM})
 is replaced by the $z$-perfect condition (2').

\begin{definition}
Define the following conditions:\\
Condition $(1)$ : For each player $i\in \mc{K}$, the private monitoring $g_i$ and
the monitoring of the mediator $m$ satisfy an $(x,y)$-coloring condition.\\
Condition $(2')$ : The channel $f$ between the mediator and each player is $z$-Perfect.
\end{definition}

\begin{theorem}[$\varepsilon$-PM]\label{theo:nodelayAPM}
Fix a strategy profile $p\in \Delta(A)$, a monitoring structure
$\M=(m,(g_i)_{i\in \mc{K}},f)$ and an $\varepsilon>0$.\\
If the monitoring structure $\M$ satisfy conditions $(1)$ and $(2')$ with,
\begin{eqnarray}
x+y+z-xy-xz-zy+xyz \leq \varepsilon
\end{eqnarray}
Then the mediator can reconstruct the $\varepsilon$-Perfect Monitoring in one-shot.
\end{theorem}

\begin{proof}
The proof is detailed in Appendix \ref{ProofTheorem4}.
\end{proof}

We provide conditions over the additive signalling structure $(\mc{W},m,f)$
that are sufficient to reconstruct the $\varepsilon$-perfect monitoring
in one-shot. Remark that a complete characterization is not available.

In order to evaluate the best additive signalling structure $(\mc{W},m,f)$,
we introduce the price of one-shot re-establishing the $\varepsilon$-Perfect Monitoring.
In the one-shot case, the entropy $H(A)$ is replaced by $\log|A|$.

\begin{definition}
Define the price of one-shot re-establishing $\varepsilon$-Perfect Monitoring:
\begin{equation}
\mathrm{PREEPM} = \frac{\log |R|}{\log |A|}
\end{equation}
\end{definition}

Finding the minimal price of one-shot re-establishing $\varepsilon$-perfect monitoring is equivalent to finding the optimal admissible additive signalling structure $(\mc{W},m,f)$.

\section{Prisoner's Dilemma}\label{SecPrisonerAPM}

We consider a simple wireless power control game where our result
 may direclty apply. Following the framework of Goodman and Mandayan
 \cite{goodman-pc-2000}, we consider a decentralized multiple access
  channel where the players choose their power control policy in order
  to maximize their energy efficiency.

We consider a two player power control game where the actions are the transmit power $p_1$ and $p_2$.
The energy-efficiency utility is defined as follows:
\begin{eqnarray}
\label{eq:def-of-utility} u_i(p_1,p_2)= \frac{f(\mathrm{SINR}_i)}{p_i} \ [\mathrm{bit} / \mathrm{J}],\qquad i\in \mc{K}
\end{eqnarray}
where the function $f(x)$ is sigmoidal (here we take $f(x)=(1-e^{-x})^M$.
The SINR at receiver $i \in \mc{K}$ writes as:
\begin{equation}
\label{eq:sinr-ne} \mathrm{SINR}_i=\frac{p_i|g_i|^2}{
p_j|g_j|^2/N+\sigma^2}
\end{equation}
with parameters $|g_i|^2$ for the channel gain and $\sigma^2$ for the noise variance.

In our previous work over this communication model \cite{LeTreustLasaulce(PowerControlRG)10}, we investigate two interesting power levels.
The first is the power of the Nash equilibrium denoted $p_i^*$ and the second is the power of the operating point $\tilde{p}_i$ which provide a Pareto optimal utility. The power $p_i^*$ and $\tilde{p}_i$ are defined respectively in equation (4) and  (9) of the article \cite{LeTreustLasaulce(PowerControlRG)10}.

In order to illustrate our results, we provide an complete analysis of a simple example which can be easily generalized.
We consider a two player power control game where only two power levels $(p_i^*,\tilde{p}_i)$ are available to the players.
Fix the parameters of the power control game for the random CDMA case. The number of players $K=2$, the number of symbols $M=2$, the spreading factor $N=2$, the channel gains $|g_1|^2=|g_2|^2=1$ and the noise variance $\sigma^2=1$. The set of achievable utility is described by the following payoff matrix. The utility pair $(0.10,0.34)$ mean that the utility of player 1 is 0.10 and 0.34 is the utility of player 2. The region of achievable utility using pure and mixed strategies is represented by the quadrilateral on figure (4).
$$\begin{tabular}{ccc}
&\phantom{b}$\tilde{p}_2$&\phantom{b}$p^*_2$\\
\end{tabular}
$$
\vspace{-0.5cm}
$$
\begin{tabular}{c|c|c|}
  \cline{2-3}
$\tilde{p}_1$&0.23,0.23 & 0.10,0.34\\
    \cline{2-3}
$p_1^*$&0.34,0.10 & 0.15,0.15\\
    \cline{2-3}
\end{tabular}
$$

\vspace{0.5cm}
Remark that this game is strategically equivalent to the Prisonner's Dilemma
where the Nash equilibrium correspond to the joint action $(p_1^*,p_2^*)$
and the social optimal action correspond to $(\tilde{p}_1,\tilde{p}_2)$.
We consider as an example a private monitoring structure $g_1,g_2$ as defined below.
We fix the additive signalling structure $(\mc{W},m,f)$, an i.i.d. mixed strategy
 $p\in \Delta(A)$ and we prove it allow the reconstruction of the $\varepsilon$-perfect
  monitoring. We compute the price of re-establishing equilibrium conditions.

The source of strategic information in this case, represents the sequence of actions of the players.
The actions are supposed to be drawn i.i.d. from a distribution over the set of the players' actions.
Denote the private monitoring of player $i$, $g_i : A \longrightarrow S_i$ with precision parameters $x'\leq x$.

\begin{figure}[!ht]
\begin{center}
\psset{xunit=0.5cm,yunit=0.6cm}
\begin{pspicture}(0,-0.5)(10,6.3)
\psdot(0,0)
\psdot(0,2)
\psdot(0,4)
\psdot(0,6)
\psdot(4,1)
\psdot(4,5)
\psdot(6,0)
\psdot(6,2)
\psdot(6,4)
\psdot(6,6)
\psdot(10,2)
\psdot(10,4)
\psline[linewidth=1pt]{->}(0,0)(4,1)
\psline[linewidth=1pt]{->}(0,2)(4,1)
\psline[linewidth=1pt]{->}(0,4)(4,5)
\psline[linewidth=1pt]{->}(0,6)(4,5)
\psline[linewidth=0.5pt,linestyle=dashed]{->}(0,0)(4,5)
\psline[linewidth=0.5pt,linestyle=dashed]{->}(0,2)(4,5)
\psline[linewidth=0.5pt,linestyle=dashed]{->}(0,4)(4,1)
\psline[linewidth=0.5pt,linestyle=dashed]{->}(0,6)(4,1)
\psline[linewidth=1pt]{->}(6,0)(10,2)
\psline[linewidth=1pt]{->}(6,2)(10,4)
\psline[linewidth=1pt]{->}(6,4)(10,2)
\psline[linewidth=1pt]{->}(6,6)(10,4)
\psline[linewidth=0.5pt,linestyle=dashed]{->}(6,0)(10,4)
\psline[linewidth=0.5pt,linestyle=dashed]{->}(6,2)(10,2)
\psline[linewidth=0.5pt,linestyle=dashed]{->}(6,4)(10,4)
\psline[linewidth=0.5pt,linestyle=dashed]{->}(6,6)(10,2)
\rput[r](-0.2,6){$\tilde{p}_1\tilde{p}_2$}
\rput[r](-0.2,4){$\tilde{p}_1p^*_2$}
\rput[r](-0.2,2){$p^*_1\tilde{p}_2$}
\rput[r](-0.2,0){$p^*_1p^*_2$}
\rput[r](5.8,6){$\tilde{p}_1\tilde{p}_2$}
\rput[r](5.8,4){$\tilde{p}_1p^*_2$}
\rput[r](5.8,2){$p^*_1\tilde{p}_2$}
\rput[r](5.8,0){$p^*_1p^*_2$}
\rput[l](4.2,5){$s_1$}
\rput[l](4.2,1){$s'_1$}
\rput[l](10.2,4){$s_2$}
\rput[l](10.2,2){$s'_2$}
\rput[u](2,5){$1-x$}
\rput[u](2,1){$1-x$}
\rput[u](8.2,5.6){$1-x'$}
\rput[u](8.2,0.3){$1-x'$}
\rput[u](6.5,3){$1-x'$}
\rput[u](2.7,2.3){$x$}
\rput[u](2.7,3.7){$x$}
\rput[u](10,3){$x'$}
\rput[u](-2,3){$g_1$}
\rput[u](4,3){$g_2$}
\end{pspicture}
\end{center}
\end{figure}

Denote $A_{g_1},A_{g_2}$ the equivalence classes of private monitoring $g_1$ and $g_2$ over the actions $A$.
\begin{eqnarray}
A_1&=&\{(\tilde{p}_1\tilde{p}_2,\tilde{p}_1p^*_2);(p^*_1\tilde{p}_2,p^*_1p^*_2)\}=\{\alpha,\alpha'\}\\
A_2&=&\{(\tilde{p}_1\tilde{p}_2,p^*_1\tilde{p}_2);(\tilde{p}_1p^*_2,p^*_1p^*_2)\}=\{\beta,\beta'\}
\end{eqnarray}
These equivalence classes induce a pair of auxiliary monitoring denoted
\begin{eqnarray}
\tilde{g_1} : A_1 &\longrightarrow& \Delta(S_1)^{|A|}\\
\alpha &\longrightarrow&  (g_1(s|a))_{a\in \alpha}\\
\tilde{g_2} : A_2 &\longrightarrow& \Delta(S_2)^{|A|}\\
\beta &\longrightarrow&  (g_2(s|a))_{a\in \beta}
\end{eqnarray}
Taking the partitions $S_{\alpha}=s_1$ ; $S_{\alpha'}=s_1'$ and $S_{\beta}=s_2$ ; $S_{\beta'}=s'_2$ we calculate the precision of the auxiliary monitoring $\tilde{g_1}$ and $\tilde{g_2}$.
\begin{eqnarray}
\min_{\alpha\in A_1}\min_{a\in \alpha}\sum_{s\in S_{\alpha}} g_1(s|a) = x\\
\min_{\beta\in A_2}\min_{a\in \beta}\sum_{s\in S_{\beta}} g_2(s|a) = x'
\end{eqnarray}
The monitoring $\tilde{g_1}$ is $x$-perfect and $\tilde{g_2}$ is $x'$-perfect.

The monitoring graphs corresponding to the above equivalence classes of the private observations are,

\begin{figure}[!ht]
\begin{center}
\psset{xunit=0.5cm,yunit=0.6cm}
\begin{pspicture}(0,-0.5)(10,2.3)
\psdot(0,0)
\psdot(0,2)
\psdot(2,0)
\psdot(2,2)
\psdot(8,0)
\psdot(8,2)
\psdot(10,0)
\psdot(10,2)
\psline[linewidth=1pt](0,2)(2,2)
\psline[linewidth=1pt](2,0)(0,0)
\psline[linewidth=1pt](8,0)(8,2)
\psline[linewidth=1pt](10,2)(10,0)
\rput[r](-0.2,2){$\tilde{p}_1\tilde{p}_2$}
\rput[l](2.2,2){$\tilde{p}_1p^*_2$}
\rput[r](-0.2,0){$p^*_1\tilde{p}_2$}
\rput[l](2.2,0){$p^*_1p^*_2$}
\rput[r](7.8,2){$\tilde{p}_1\tilde{p}_2$}
\rput[l](10.2,2){$\tilde{p}_1p^*_2$}
\rput[r](7.8,0){$p^*_1\tilde{p}_2$}
\rput[l](10.2,0){$p^*_1p^*_2$}
\rput[u](-2,1){$\G_{\tilde{g_1}}$}
\rput[u](6,1){$\G_{\tilde{g_2}}$}
\end{pspicture}
\end{center}
\end{figure}

Note that by construction, each of those graph is an union of complete graphs.

The mediator observes a signal drawn from $m : A \longrightarrow Q$.

\begin{figure}[!ht]
\begin{center}
\psset{xunit=0.5cm,yunit=0.5cm}
\begin{pspicture}(-1,-0.5)(4.5,6.3)
\psdot(0,0)
\psdot(0,2)
\psdot(0,4)
\psdot(0,6)
\psdot(4,0)
\psdot(4,6)
\psdot(4,3)
\psline[linewidth=1pt]{->}(0,0)(4,0)
\psline[linewidth=1pt]{->}(0,2)(4,3)
\psline[linewidth=1pt]{->}(0,4)(4,3)
\psline[linewidth=1pt]{->}(0,6)(4,6)
\psline[linewidth=0.5pt,linestyle=dashed]{->}(0,0)(4,3)
\psline[linewidth=0.5pt,linestyle=dashed]{->}(0,2)(4,0)
\psline[linewidth=0.5pt,linestyle=dashed]{->}(0,4)(4,6)
\psline[linewidth=0.5pt,linestyle=dashed]{->}(0,6)(4,3)
\rput[r](-0.2,6){$\tilde{p}_1\tilde{p}_2$}
\rput[r](-0.2,4){$\tilde{p}_1p^*_2$}
\rput[r](-0.2,2){$p^*_1\tilde{p}_2$}
\rput[r](-0.2,0){$p^*_1p^*_2$}
\rput[l](4.2,6){$q_1$}
\rput[l](4.2,3){$q_2$}
\rput[l](4.2,0){$q_3$}
\rput[u](2,6.4){$1-y$}
\rput[u](2,-0.6){$1-y$}
\rput[u](1.1,3){$1-y$}
\rput[u](2.4,1.1){$y$}
\rput[u](2.4,4.9){$y$}
\rput[u](-2,3){$m$}
\end{pspicture}
\end{center}
\end{figure}

Denote the equivalence classes of the monitoring $m$ of the mediator over the actions $A$.
\begin{eqnarray}
A_m&=&\{(\tilde{p}_1\tilde{p}_2);(\tilde{p}_1p^*_2,p^*_1\tilde{p}_2);(p^*_1p^*_2))\}=\{\gamma,\gamma',\gamma''\}
\end{eqnarray}
These equivalence classes induce an auxiliary monitoring denoted
\begin{eqnarray}
\tilde{m} : A_m &\longrightarrow& \Delta(Q)^{|A|}\\
\gamma &\longrightarrow&  (m(q|a))_{a\in \gamma}
\end{eqnarray}
Taking the partitions $Q_{\gamma}=q_1$, $Q_{\gamma'}=q_2$ and $Q_{\gamma''}=q_3$ we calculate the precision of the auxiliary monitoring $\tilde{m}$.
\begin{eqnarray}
\min_{\gamma\in A_m}\min_{a\in \gamma}\sum_{q\in Q_{\gamma}} m(q|a) = y
\end{eqnarray}
The monitoring $\tilde{m}$ is $y$-perfect.

In order to decide whether the mediator can reconstruct the desired monitoring,
let us check if the auxiliary monitoring $\tilde{m}$ of the mediator is a coloring
of the graphs $\G_{\tilde{g_1}}$ and $\G_{\tilde{g_2}}$ of the players.
To illustrate this, we associate the colors blue, red and green to respectively $q_1$, $q_2$ and $q_3$.
\begin{figure}[!ht]
\begin{center}
\psset{xunit=0.5cm,yunit=0.6cm}
\begin{pspicture}(0,-0.5)(10,2.3)
\psdot(0,0)
\psdot(0,2)
\psdot(2,0)
\psdot(2,2)
\psdot(8,0)
\psdot(8,2)
\psdot(10,0)
\psdot(10,2)
\pscircle[linewidth=1pt,linecolor=red](0,0){0.1}
\pscircle[linewidth=1pt,linecolor=blue](0,2){0.1}
\pscircle[linewidth=1pt,linecolor=green](2,0){0.1}
\pscircle[linewidth=1pt,linecolor=red](2,2){0.1}
\pscircle[linewidth=1pt,linecolor=red](8,0){0.1}
\pscircle[linewidth=1pt,linecolor=blue](8,2){0.1}
\pscircle[linewidth=1pt,linecolor=green](10,0){0.1}
\pscircle[linewidth=1pt,linecolor=red](10,2){0.1}
\psline[linewidth=1pt](0,2)(2,2)
\psline[linewidth=1pt](2,0)(0,0)
\psline[linewidth=1pt](8,0)(8,2)
\psline[linewidth=1pt](10,2)(10,0)
\rput[r](-0.2,2){$\tilde{p}_1\tilde{p}_2$}
\rput[l](2.2,2){$\tilde{p}_1p^*_2$}
\rput[r](-0.2,0){$p^*_1\tilde{p}_2$}
\rput[l](2.2,0){$p^*_1p^*_2$}
\rput[r](7.8,2){$\tilde{p}_1\tilde{p}_2$}
\rput[l](10.2,2){$\tilde{p}_1p^*_2$}
\rput[r](7.8,0){$p^*_1\tilde{p}_2$}
\rput[l](10.2,0){$p^*_1p^*_2$}
\rput[u](-0.2,2.4){$\textcolor[rgb]{0.00,0.00,1.00}{q_1}$}
\rput[u](2.2,2.4){$\textcolor[rgb]{1.00,0.00,0.00}{q_2}$}
\rput[u](-0.2,-0.4){$\textcolor[rgb]{1.00,0.00,0.00}{q_2}$}
\rput[u](2.2,-0.4){$\textcolor[rgb]{0.00,1.00,0.00}{q_3}$}
\rput[u](7.8,2.4){$\textcolor[rgb]{0.00,0.00,1.00}{q_1}$}
\rput[u](10.2,2.4){$\textcolor[rgb]{1.00,0.00,0.00}{q_2}$}
\rput[u](7.8,-0.4){$\textcolor[rgb]{1.00,0.00,0.00}{q_2}$}
\rput[u](10.2,-0.4){$\textcolor[rgb]{0.00,1.00,0.00}{q_3}$}
\rput[u](-2,1){$\G_{\tilde{g_1}}$}
\rput[u](6,1){$\G_{\tilde{g_2}}$}
\end{pspicture}
\end{center}
\end{figure}
For each player $i\in \mc{K}$, the pair of auxiliary monitoring $(m,g_i)$ satisfy an $(x,y)$-coloring condition (recall that $x'\leq x$).
Thus, the mediator gets sufficient information to reconstruct the $\varepsilon$-perfect monitoring at the players with $\varepsilon = x+y-xy$.

To extract the essential information from mediator's signal $q$ without decreasing the precision of the monitoring, let us introduce the following bi-auxiliary graph.
\begin{figure}[!ht]
\begin{center}
\psset{xunit=0.6cm,yunit=0.6cm}
\begin{pspicture}(-0.4,-0.5)(2.4,2.3)
\psdot(0,0)
\psdot(0,2)
\psdot(1.7,1)
\pscircle[linewidth=1pt,linecolor=cyan](0,0){0.1}
\pscircle[linewidth=1pt,linecolor=magenta](0,2){0.1}
\pscircle[linewidth=1pt,linecolor=magenta](1.7,1){0.1}
\psline[linewidth=1pt](0,0)(0,2)
\psline[linewidth=1pt](0,0)(1.7,1)
\rput[r](-0.3,0){$q_2$}
\rput[r](-0.3,2){$q_1$}
\rput[l](2,1){$q_3$}
\rput[u](0.1,2.5){$\textcolor[rgb]{1.00,0.00,0.50}{r_1}$}
\rput[u](0.1,-0.5){$\textcolor[rgb]{0.00,1.00,1.00}{r_2}$}
\rput[u](1.6,1.5){$\textcolor[rgb]{1.00,0.00,0.50}{r_1}$}
\rput[u](-2,1){$\G_{m}$}
\end{pspicture}
\end{center}
\end{figure}
The coloring $h : Q \longrightarrow R$ of the above bi-auxiliary graph characterizes the essential information the mediator should give to the player in order to re-establish the $\varepsilon$-perfect monitoring. Recall that this essential information is optimal in the sense of the cardinality of $R$ and of the precision of the monitoring. It cannot be reduced without introducing a larger ambiguity between action profiles for at least one player. The process of strategic information is described as follows:

\begin{figure}[!ht]
\begin{center}
\psset{xunit=0.5cm,yunit=0.6cm}
\begin{pspicture}(0,-0.5)(10,6.3)
\psdot(0,0)
\psdot(0,2)
\psdot(0,4)
\psdot(0,6)
\psdot(4,0)
\psdot(4,3)
\psdot(4,6)
\psdot(6,0)
\psdot(6,3)
\psdot(6,6)
\psdot(10,1)
\psdot(10,4)
\psline[linewidth=1pt]{->}(0,0)(4,0)
\psline[linewidth=1pt]{->}(0,2)(4,3)
\psline[linewidth=1pt]{->}(0,4)(4,3)
\psline[linewidth=1pt]{->}(0,6)(4,6)
\psline[linewidth=1pt]{->}(6,0)(10,4)
\psline[linewidth=1pt]{->}(6,3)(10,1)
\psline[linewidth=1pt]{->}(6,6)(10,4)
\rput[r](-0.2,6){$\tilde{p}_1\tilde{p}_2$}
\rput[r](-0.2,4){$\tilde{p}_1p^*_2$}
\rput[r](-0.2,2){$p^*_1\tilde{p}_2$}
\rput[r](-0.2,0){$p^*_1p^*_2$}
\rput[l](4.2,6){$q_1$}
\rput[l](4.2,3){$q_2$}
\rput[l](4.2,0){$q_3$}
\rput[r](5.8,6){$q_1$}
\rput[r](5.8,3){$q_2$}
\rput[r](5.8,0){$q_3$}
\rput[l](10.2,4){$r_1$}
\rput[l](10.2,1){$r_2$}
\rput[u](-2,3){$m$}
\psline[linewidth=0.5pt,linestyle=dashed]{->}(0,0)(4,3)
\psline[linewidth=0.5pt,linestyle=dashed]{->}(0,2)(4,0)
\psline[linewidth=0.5pt,linestyle=dashed]{->}(0,4)(4,6)
\psline[linewidth=0.5pt,linestyle=dashed]{->}(0,6)(4,3)
\rput[u](2,6.4){$1-y$}
\rput[u](2,-0.6){$1-y$}
\rput[u](1.1,3){$1-y$}
\rput[u](2.4,1.1){$y$}
\rput[u](2.4,4.9){$y$}
\rput[u](5,5){$h$}
\end{pspicture}
\end{center}
\end{figure}
For each player $i\in \mc{K}$, the pair of auxiliary monitoring $(h\circ m,g_i)$ still satisfy an $(x,y)$-coloring condition and is moreover minimal in term of cardinality $|R|$. To reconstruct the $\varepsilon$-perfect monitoring at the player, the mediator will send the common information $r$ such that each player, knowing the private signal $s_i$ can reconstruct the right action profile with probability more than $1-\varepsilon$.

Suppose now that player 1 plays a mixed strategy $(2/3,1/3)$ and player 2 plays a mixed strategy $(2/3,1/3)$. Assume from now that they play repeatedly  following this mixed strategy. A sequence of action profiles is generated from the distribution $p\in \Delta(A)$ and it leads to the payoff vector $(0.22,0.22)$ (see Fig. (4)).
$$\begin{tabular}{ccc}
&\phantom{b}$\tilde{p}_2$&\phantom{b}$p^*_2$\\
\end{tabular}
$$
\vspace{-0.6cm}
$$
\begin{tabular}{c|c|c|}
  \cline{2-3}
$\tilde{p}_1$&4/9 & 2/9\\
    \cline{2-3}
$p^*_1$&2/9 & 1/9\\
    \cline{2-3}
\end{tabular}
$$
The entropy of such a distribution source is $H(a)=\log 9 -4/3\simeq 1.8366$. Fix the noise level of the transitions functions at $x=x'=y=1/10$. The process of information generated a source of essential information $(r_1,r_2)$ with distribution $(49/90, 41/90)$ of entropy $H(r) \simeq 0.9943$.
To transmit the source of essential information, the mediator considers the side information $s_i$ from the transition channel of player $i$.
\begin{eqnarray}
T_i(s|r_1) &=& \frac{\sum_{a,q}\PPP(a,q,r,s)}{\sum_{a,q}\PPP(a,q,r)} \\
& =& \frac{\sum_{a,q}p(a)m(q|a)h(r|q)g_i(s|a)}{\sum_{a,q}\sum_{a,q}p(a)m(q|a)h(r|q)}
\end{eqnarray}
The transition matrix of the channel are evaluated.
\begin{eqnarray}
T_1(s_1|r_1) &=& 353/490 \\
T_1(s'_1|r_1)&=& 137 /490 \\
T_1(s_1|r_2) &=& 217/410 \\
T_1(s'_1|r_2) &=& 193/410
\end{eqnarray}
We represent it as a binary channel.
\begin{figure}[!ht]
\begin{center}
\psset{xunit=1cm,yunit=1cm}
\begin{pspicture}(0,0)(4,2.3)
\psdot(0,0)
\psdot(0,2)
\psdot(4,0)
\psdot(4,2)
\psline[linewidth=1pt]{->}(0,0)(4,0)
\psline[linewidth=1pt]{->}(0,2)(4,2)
\psline[linewidth=1pt]{->}(0,0)(4,2)
\psline[linewidth=1pt]{->}(0,2)(4,0)
\rput[r](-0.2,0){$r_2$}
\rput[r](-0.2,2){$r_1$}
\rput[l](4.2,0){$s'_1$}
\rput[l](4.2,2){$s_1$}
\rput[u](2.4,2.2){$353/490$}
\rput[d](1.8,-0.3){$217/410$}
\rput[u](3.1,0.9){$137/490$}
\rput[u](0.8,0.8){$193/410$}
\rput[u](-1,1){$T_1$}
\end{pspicture}
\end{center}
\end{figure}
The channel transition of player 2 is also characterized and is found to be identical.
Using the Slepian and Wolf binning scheme, the entropy of the essential information with side information writes as :
\begin{eqnarray}
H=\max_{i\in \mc{K}}H(R|S_i)
\end{eqnarray}
In this case, we have $p(r_1|s_1)= 353/570$, $p(r_2|s'_1)=217/330$, $p(r_2|s_1)=193/570$ and $p(r_1|s'_1)=137/330$ and $T_1=T_2$. The minimal information rate sent by the mediator to both players is
\begin{eqnarray}
H=H(R|S_1)=H(R|S_2)\simeq 0.9451
\end{eqnarray}
Under the condition that the rates pair $(H,H)$ belong to the capacity region of the channel between the mediator and the players, the mediator can reconstruct the $\varepsilon$-perfect monitoring at the players with precision $\varepsilon= x+y-xy = 19/100$. \\
The price of re-establishing $\varepsilon$-Perfect Monitoring writes:
\begin{equation}
\mathrm{PREEPM}^{\infty} = \frac{\max_{i\in \mc{K}}H(r|s_i)}{H(a)} \simeq \frac{0.9451}{1.8366}\simeq 0.5145
\end{equation}
Taking the same monitoring structure with a noise level $x=x'=y=0$,
we investigate the noiseless version of the reconstruction of the perfect monitoring.
The price of re-establishing Perfect Monitoring becomes:
\begin{equation}
\mathrm{PRPM}^{\infty} = \frac{\max_{i\in \mc{K}}H(r|s_i)}{H(a)} \simeq 0.5
\end{equation}
We conclude that in the noiseless problem, the private
 monitoring structure provides almost half the information
needed by the players to reconstruct the source of strategic information.
Whereas in the noisy case, the additional monitoring structure is in charge
 by almost 51.5 \% of the reconstruction the strategic source of  information.

\section{Concluding remarks}

The monitoring problem of strategic information is addressed
in this paper. Taking into account the private monitoring structure, a mediator is introduced in order to re-establish
$\varepsilon$-perfect monitoring at the players. In order to
evaluate the signaling cost for the mediator, the problem of
the reconstructing a strategic information is re-interpreted as a channel of communication
theory. Graph theory and Shannon theory are respectively exploited
to provide a characterization of the admissible monitoring structure
and analyze their efficiency in term of ``price of re-establishing
$\varepsilon$-Perfect Monitoring" ($PREEPM$). A coding theorem
is provided for the channels where the mediator observes the source
imperfectly and the strategic information is drawn from an i.i.d.
source. Challenging open problems appear when considering a source of information
 generated by an arbitrary stochastic process. For example, in the case of imperfect monitoring of past actions,
 the players can choose an appropriate sequence of actions such as to manipulate the coding schemes.
Another interesting extension is to consider a mediator that sends private messages
to the players instead of common messages.
It would also be of interest to provide conditions for changing an imperfect monitoring structure $\M$
into another imperfect monitoring structure $\M'$ (not necessarily perfect or almost perfect).

\bibliographystyle{plain}
\bibliography{BiblioMael}

\begin{thebibliography}{10}

\bibitem{BasarOlsder82}
T.~Basar and G.J. Olsder.
\newblock {\em Dynamic noncooperative game theory}.

\bibitem{BondyMurty}
J.A. Bondy and U.S.R. Murty.
\newblock {\em Graph theory with applications}.
\newblock Elsevier Science Publishing Co., 1976.

\bibitem{Brown51}
G.W. Brown.
\newblock Iterative solution of games by fictitious play.
\newblock In {\em Activity {A}nalysis of {P}roduction and {A}llocation}, Cowles
  Commission Monograph No. 13, pages 374--376. John Wiley \& Sons Inc., New
  York, N. Y., 1951.

\bibitem{cover-book-2006}
T.~M. Cover and J.~A. Thomas.
\newblock {\em Elements of information theory}.
\newblock 2nd. Ed., Wiley-Interscience, New York, 2006.

\bibitem{ElyValimaki(Robust)02}
J.C. Ely and J.~Valimaki.
\newblock A robust folk theorem for the prisoner's dilemma.
\newblock {\em Journal of Economic Theory}, 102(1):84--105, January 2002.

\bibitem{FudenbergYamamoto09}
D.~Fudenberg and Y.~Yamamoto.
\newblock The folk theorem for irreducible stochastic games with imperfect
  public monitoring.
\newblock {\em Journal of Economic Theory}, 146(4).

\bibitem{goodman-pc-2000}
D.~J. Goodman and N.~B. Mandayam.
\newblock Power control for wireless data.
\newblock {\em IEEE Person. Comm.}, 7:48--54, 2000.

\bibitem{GossnerHernandezNeyman06}
O.~Gossner, P.~Hernandez, and A.~Neyman.
\newblock Optimal use of communication resources.
\newblock {\em Econometrica}, 74(6):1603--1636, 2006.

\bibitem{HornerOlszewski06}
J.~H\"{o}rner and W.~Olszewski.
\newblock The folk theorem for games with private almost-perfect monitoring.
\newblock {\em Econometrica}, 74(6):1499--1544, 2006.

\bibitem{korner-it-1977}
J.~Körner and K.~Marton.
\newblock General broadcast channels with degraded message sets.
\newblock IT-23:60--64, Jan. 1977.

\bibitem{Lasaulce-Tutorial-09}
S.~Lasaulce, M.~Debbah, and E.~Altman.
\newblock Methodologies for analyzing equilibria in wireless games.
\newblock {\em IEEE Signal Processing Magazine, Special issue on Game Theory
  for Signal Processing}, Sep. 2009.

\bibitem{LeTreustLasaulce(PowerControlRG)10}
M.~LeTreust and S.~Lasaulce.
\newblock A repeated game formulation of energy-efficient decentralized power
  control.
\newblock {\em IEEE Trans. on Wireless Commun.}, 9(9):2860 -- 2869, Sept. 2010.

\bibitem{MerhavShamai03}
N.~Merhav and S.~Shamai.
\newblock On joint source-channel coding for the wyner-ziv source and the
  gel'fand-pinsker channel.
\newblock {\em IEEE Transactions on Information Theory}, 49(11):2844--2855, Nov
  2003.

\bibitem{RT98}
J.~Renault and T.~Tomala.
\newblock Repeated proximity games.
\newblock {\em International Journal of Games Theory}, 27:539--559, 1998.

\bibitem{RT04}
J.~Renault and T.~Tomala.
\newblock Communication equilibrium payoffs in repeated games with imperfect
  monitoring.
\newblock {\em Games and Economics Behaviour}, 49:313--344, 2004.

\bibitem{Sastry-94}
P.S. Sastry, V.V. Phansalkar, and M.A.L. Thathachar.
\newblock Decentralized learning of {N}ash equilibria in multi-person
  stochastic games with incomplete information.
\newblock {\em IEEE Transactions on Systems, Man and Cybernetics},
  24(5):769--777, May 1994.

\bibitem{shannon-bell-1948}
C.~E. Shannon.
\newblock A mathematical theory of communication.
\newblock {\em Bell System Technical Journal}, 27:379--423, 1948.

\bibitem{Sorin92}
S.~Sorin.
\newblock {\em Repeated Games with Complete Information, in Hanbook of Game
  Theory with Economic Applications}, volume~1.
\newblock Elsevier Science Publishers, 1992.

\bibitem{Tomala98}
T.~Tomala.
\newblock Pure equilibria of repeated games with public observation.
\newblock {\em International Journal of Game Theory}, 27(1):93--109, 1998.

\bibitem{Chandramouli-08}
Yiping Xing and R.~Chandramouli.
\newblock Stochastic learning solution for distributed discrete power control
  game in wireless data networks.
\newblock {\em IEEE/ACM Trans. Networking}, 16(4):932--944, 2008.

\end{thebibliography}

\appendix \label{Appendix}

\section{Proof of Proposition 1}\label{ProofProposition1}

By definition, $\varepsilon$ is the minimum admissible value such that:
\begin{tiny}
\begin{eqnarray*}
&&\exists \;T=(T_a)_a,\; \forall a\in A,\quad \sum_{\sigma_i\in T_a}\Lambda(\sigma_i|a)\geq 1-\varepsilon\\
&\Longleftrightarrow& \exists \;T=(T_a)_a,\quad  \min_{a\in A} \sum_{\sigma_i\in T_a}\Lambda(\sigma_i|a)\geq 1-\varepsilon\\
&\Longleftrightarrow& 1-\varepsilon\leq\max_{T=(T_a)_a}\min_{a\in A}\sum_{\sigma_i\in T_a} \Lambda(\sigma_i|a)
\end{eqnarray*}
\end{tiny}
Taking the minimum admissible value for $\varepsilon$, the monitoring $\Lambda$ is $\varepsilon$-perfect if and only if there is equality in the above equation.

\section{Proof of Theorem 2}\label{ProofTheorem2}

\begin{proof}
\begin{footnotesize}
We will prove that the conditions $(1)$ and $(2)$ are sufficient. The first conditions $(1)$ states that there exists a pair $(x,y)$ such that $x+y-xy\leq \varepsilon$ and for each player $i\in \mc{K}$, the private monitoring $g_i$ and the monitoring of the mediator $m$ satisfy an $(x,y)$-coloring condition. Taking now the minimal coloring of the bi-auxiliary graph $\tilde{h} :Q \longrightarrow R$, we will show that for every player $i\in \mc{K}$, the joint monitoring $(g_i,h\circ m)$ is $x+y-xy$ perfect.\\
Fix a player $i\in \mc{K}$. Let $\{Q_{\beta}\}_{\beta\in A_m}$ be the partition of signals $q\in Q$ indexed by the equivalence classes of $A$ with respect to the monitoring $m$ and $\{S_{\alpha}\}_{\alpha\in A_{g_i}}$ the partition of signals $s\in S_i$ indexed by the equivalence classes of $A$ with respect to the monitoring $g_i$. By hypothesis, $\{Q_{\beta}\}_{\beta\in A_m}$ is a coloring of the graph $G_{g_i}$. We first show that $(\tilde{h}(Q_{\beta}))_{\beta \in A_m}$ is still a coloring of the graph $G_{g_i}$. Take $a$ and $b$ two neighbor nodes of the graph $G_{g_i}$. By the coloring property, the sets of associated color $Q_{\alpha}(a)$ and $Q_{\alpha}(b)$ are disjoint. Thus each pair of colors $q\in Q_{\alpha}(a)$ and  $q'\in Q_{\alpha}(b)$ are neighbor in the bi-auxiliary graph. The coloring $\tilde{h} :Q \longrightarrow R$ of the bi-auxiliary graph implies that $(\tilde{h}(Q_{\alpha}))_{\alpha\in A_m}$ is still a coloring of the graph $G_{g_i}$.
Second, the coloring property implies that the following product $T_a=S_{\alpha}(a) \times \tilde{h}(Q_{\alpha}(a))$ defines a partition $(T_a)_{a\in A}$ of $T$. Let us calculate the precision of such a joint monitoring. For all strategic information $a\in A$:
\begin{tiny}
\begin{eqnarray*}
\sum_{s,r\in T_a}\Lambda(s,r|a)&=&\sum_{s\in S_{\alpha}(a)}g_i(s|a)\sum_{q\in Q_{\beta}(a)}\sum_{r\in \tilde{h}(Q_{\beta}(a))}\tilde{h}(r|q)m(q|a)\\
&\geq&\sum_{s\in S_{\alpha}(a)} g(r|a)\sum_{q\in Q_{\beta}(a)}m(q|a)\\
&\geq& (1-y)(1-x)=1-(x+y-xy)
\end{eqnarray*}
\end{tiny}
Thus, for each player $i\in \mc{K}$, the monitoring $(g_i,\tilde{h}\circ m)$ satisfies an $x+y-xy$-Perfect Monitoring condition. It remains to transmit that signal over the broadcast channel with common messages $f$. The condition $(2)$ states that the essential rate $H$ satisfies $H\leq \C_0$, the capacity $\C_0$ of the broadcast channel $f$ with common messages \cite{korner-it-1977}. The joint source-channel coding theorem states that there exists appropriate mappings:
\begin{eqnarray}
&\phi& :  R ^n \longrightarrow X^n\\
&\psi_i& : Y_i^n \times S_i^n \longrightarrow A^n,\qquad \forall i\in \mc{K}
\end{eqnarray}
such that transmitting the source $r$ over the broadcast channel with common messages $f$ is possible with an error probability $\PP_e^n\leq \delta$.
The above mappings correctly transmit every sequence $r^n$ with probability more than $\PP(r^n=\hat{r^n})\geq 1-\delta$. When the sequence $r^n$ is correctly decoded, at each stages the symbol $r$ combined with the side symbol $s_i$ for each player $i$, are associated to an strategic information profile $a$ where the stage error probability is bounded by $x+y-xy$. We proved that
\begin{tiny}
\begin{eqnarray*}
\PP\left[\exists T_i=\{T_i^a: a \in A\},\;\forall a\in A,\; \sum_{\sigma_i\in T_i^a} \Lambda(\sigma_i|a)\geq 1-\varepsilon\right]\geq 1-\delta
\end{eqnarray*}
\end{tiny}
\end{footnotesize}
\end{proof}

\section{Proof of Theorem 3}\label{ProofTheorem3}

\begin{proof}
\begin{footnotesize}
First we show that conditions $(1')$ and $(2)$ are sufficient.
The monitoring of the mediator $m$ is a painting of the family of graphs $(\G_i)_{i\in \mc{K}}$.
This implies that for each player $i \in \mc{K}$, for each pair of strategic information $a,b\in A$,
 if the private signal as a positive probability to be the same, then the signal observed
 by the mediator will distinguish them.
\begin{eqnarray}
g_i(a)\cap g_i(b)\neq \emptyset \Longrightarrow m(a)\cap m(b) = \emptyset
\end{eqnarray}
Moreover, the recoloring $h :Q \longrightarrow R$ keeps this property. For all player $i\in \mc{K}$,
\begin{eqnarray}
g_i(a)\cap g_i(b)\neq \emptyset&& \\
q\in m(a),\; q'\in m(b) &&\Longrightarrow h(q)\neq h(q') \Longrightarrow r\neq r'
\end{eqnarray}
Condition $(1')$ implies that, for each player $i\in \mc{K}$, the pair of information $(s_i,r)$ is sufficient to reconstruct the Perfect Monitoring.\\
Condition $(2)$ states that $H\leq \C_0$ which implies that the rate of this information $r$ is lower than the capacity of the channel between the mediator and the player $i\in \mc{K}$. Thus, by the source-channel coding theorem, we have that:
\begin{eqnarray}
\forall \varepsilon>0,\;\exists (n,h,\phi,(\psi_i)_{i\in \mc{K}})\text{-process such that},\quad \PP^n_e\leq \varepsilon
\end{eqnarray}
This implies that the mediator can reconstruct the Perfect Monitoring.
Second, we show that conditions $(1')$ and $(2)$ are necessary.
Suppose that condition $(2)$ does not hold. Then, by the source-channel coding theorem of Merhav and Shamai (2003 \cite{MerhavShamai03}), it is impossible to transmit the source $r$ over the channel $f$ with low error probability. The mediator cannot reconstruct the Perfect Monitoring.\\
Suppose that condition $(1')$ does not hold. Then, there exists a player $i$ and a pair of strategic information $a,b$ that have the same color $s$ and  there exists an edge $e=(a,b)$. This implies that with positive probability player $i$ will observe a private signal $r$ and a public signal $s$ when strategic information $a$ or $b$ is drawn. Then, the mediator cannot reconstruct the Perfect Monitoring.
\end{footnotesize}
\end{proof}

\section{Proof of Theorem 4}\label{ProofTheorem4}

\begin{proof}
\begin{footnotesize}
We ever show that the condition $(1)$ is sufficient to reconstruct $x+y-xy$-Perfect Monitoring in one shot. We show that if the family of channels $(f_i)_{i\in \mc{K}}$ between the mediator and each player satisfy an $z$-perfect condition (condition (2')), then each player monitors with a precision at least of $x+y+z-xy-xz-yz+xyz$.
Let us calculate the precision of such a joint monitoring received by player $i$ is $\Lambda_i : A \longrightarrow  S_i  \times Y_i$. For all joint strategic information $a\in A$ and for each player $i\in \mc{K}$, the $(x,y)$-coloring property guarantees the existence of a partition defined by $T_a^i= S_{\alpha}^i\times Q_{\beta}\times Y^r_i$ such that $a\in \alpha$, $a\in \beta$ and $r\in \tilde{h}(Q_{\beta})$.  The precision of the joint monitoring is upper bounded by,
\begin{tiny}
\begin{eqnarray*}
&&\sum_{s,y\in T_a}\Lambda_i(s,y|a)\\
&=&\sum_{s\in S_{\alpha}^i}g_i(s|a)\sum_{q\in Q_{\beta}}\sum_{r\in \tilde{h}(Q_{\beta})}\sum_{y\in Y^r_i}f_i(y|r)\tilde{h}(r|q)m(q|a)\\
&\geq& \sum_{s\in S_{\alpha}^i}g_i(s|a)\sum_{q\in Q_{\beta}}m(q|a)\sum_{r\in \tilde{h}(Q_{\beta})}\tilde{h}(r|q)\sum_{y\in Y^r_i}f_i(y|r)\\
&=&\sum_{s\in S_{\alpha}^i}g_i(s|a)\sum_{q\in Q_{\beta}}m(q|a)\sum_{y\in Y^r_i}f_i(y|r)\\
&\geq& (1-y)(1-x)(1-z)\\
&=&1-(x+y+z-xy-xz-zy+xyz)
\end{eqnarray*}
\end{tiny}
For each player $i\in \mc{K}$, the joint monitoring $\Lambda_i=(g_i,f_i\circ \tilde{h} \circ m)$ is $x+y+z-xy-xz-yz+xyz$ perfect.
\end{footnotesize}
\end{proof}

\begin{figure}[hb]\label{fig:MatrixPowerGame}
	\centering
		\includegraphics[width=0.45\textwidth]{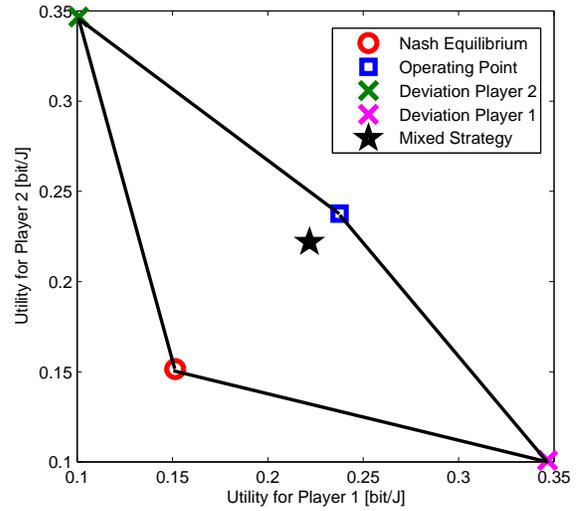}
\vspace{-0.5cm}
	\caption{Nash Equilibrium, Operating Point and the Deviation Utilities for $(K,M,N ) = (2,2,2)$}
\end{figure}

\end{document}